\documentclass{article}
\usepackage[utf8]{inputenc}
\usepackage[T1]{fontenc}
\usepackage{fullpage}
\usepackage{fancybox}
\usepackage{amsmath}
\usepackage{amscd}
\usepackage{moreverb}
\usepackage{commath}
\usepackage[ruled,vlined]{algorithm2e}
\usepackage{listings}
\usepackage[standard]{ntheorem}

\usepackage{dsfont}

\usepackage{stmaryrd}

\usepackage{graphicx}
\usepackage{subfigure}

\usepackage{color}

\title{Efficient and Cryptographically Secure Generation of Chaotic Pseudorandom Numbers on GPU}
\begin{document}

\author{Jacques M. Bahi, Rapha\"{e}l Couturier,  Christophe
Guyeux, and Pierre-Cyrille Héam\thanks{Authors in alphabetic order}}
   
\maketitle

\begin{abstract}
In this paper we present a new pseudorandom number generator (PRNG) on
graphics processing units  (GPU). This PRNG is based  on the so-called chaotic iterations.  It
is firstly proven  to be chaotic according to the Devaney's  formulation. We thus propose  an efficient
implementation  for  GPU that successfully passes the   {\it BigCrush} tests, deemed to be the  hardest
battery of tests in TestU01.  Experiments show that this PRNG can generate
about 20 billion of random numbers  per second on Tesla C1060 and NVidia GTX280
cards.
It is then established that, under reasonable assumptions, the proposed PRNG can be cryptographically 
secure.
A chaotic version of the Blum-Goldwasser asymmetric key encryption scheme is finally proposed.

\end{abstract}

\section{Introduction}

Randomness is of importance in many fields such as scientific simulations or cryptography. 
``Random numbers'' can mainly be generated either by a deterministic and reproducible algorithm
called a pseudorandom number generator (PRNG), or by a physical non-deterministic 
process having all the characteristics of a random noise, called a truly random number
generator (TRNG). 
In this paper, we focus on reproducible generators, useful for instance in
Monte-Carlo based simulators or in several cryptographic schemes.
These domains need PRNGs that are statistically irreproachable. 
In some fields such as in numerical simulations, speed is a strong requirement
that is usually attained by using parallel architectures. In that case,
a recurrent problem is that a deflation of the statistical qualities is often
reported, when the parallelization of a good PRNG is realized.
This is why ad-hoc PRNGs for each possible architecture must be found to
achieve both speed and randomness.
On the other side, speed is not the main requirement in cryptography: the great
need is to define \emph{secure} generators able to withstand malicious
attacks. Roughly speaking, an attacker should not be able in practice to make 
the distinction between numbers obtained with the secure generator and a true random
sequence. 
Finally, a small part of the community working in this domain focuses on a
third requirement, that is to define chaotic generators.
The main idea is to take benefits from a chaotic dynamical system to obtain a
generator that is unpredictable, disordered, sensible to its seed, or in other word chaotic.
Their desire is to map a given chaotic dynamics into a sequence that seems random 
and unassailable due to chaos.
However, the chaotic maps used as a pattern are defined in the real line 
whereas computers deal with finite precision numbers.
This distortion leads to a deflation of both chaotic properties and speed.
Furthermore, authors of such chaotic generators often claim their PRNG
as secure due to their chaos properties, but there is no obvious relation
between chaos and security as it is understood in cryptography.
This is why the use of chaos for PRNG still remains marginal and disputable.

The authors' opinion is that topological properties of disorder, as they are
properly defined in the mathematical theory of chaos, can reinforce the quality
of a PRNG. But they are not substitutable for security or statistical perfection.
Indeed, to the authors' mind, such properties can be useful in the two following situations. On the
one hand, a post-treatment based on a chaotic dynamical system can be applied
to a PRNG statistically deflective, in order to improve its statistical 
properties. Such an improvement can be found, for instance, in~\cite{bgw09:ip,bcgr11:ip}.
On the other hand, chaos can be added to a fast, statistically perfect PRNG and/or a
cryptographically secure one, in case where chaos can be of interest,
\emph{only if these last properties are not lost during
the proposed post-treatment}. Such an assumption is behind this research work.
It leads to the attempts to define a 
family of PRNGs that are chaotic while being fast and statistically perfect,
or cryptographically secure.
Let us finish this paragraph by noticing that, in this paper, 
statistical perfection refers to the ability to pass the whole 
{\it BigCrush} battery of tests, which is widely considered as the most
stringent statistical evaluation of a sequence claimed as random.
This battery can be found in the well-known TestU01 package~\cite{LEcuyerS07}.
Chaos, for its part, refers to the well-established definition of a
chaotic dynamical system proposed by Devaney~\cite{Devaney}.

In a previous work~\cite{bgw09:ip,guyeux10} we have proposed a post-treatment on PRNGs making them behave
as a chaotic dynamical system. Such a post-treatment leads to a new category of
PRNGs. We have shown that proofs of Devaney's chaos can be established for this
family, and that the sequence obtained after this post-treatment can pass the
NIST~\cite{Nist10}, DieHARD~\cite{Marsaglia1996}, and TestU01~\cite{LEcuyerS07} batteries of tests, even if the inputted generators
cannot.
The proposition of this paper is to improve widely the speed of the formerly
proposed generator, without any lack of chaos or statistical properties.
In particular, a version of this PRNG on graphics processing units (GPU)
is proposed.
Although GPU was initially designed  to accelerate
the manipulation of  images, they are nowadays commonly  used in many scientific
applications. Therefore,  it is important  to be able to  generate pseudorandom
numbers inside a GPU when a scientific application runs in it. This remark
motivates our proposal of a chaotic and statistically perfect PRNG for GPU.  
Such device
allows us to generate almost 20 billion of pseudorandom numbers per second.
Furthermore, we show that the proposed post-treatment preserves the
cryptographical security of the inputted PRNG, when this last has such a 
property.
Last, but not least, we propose a rewriting of the Blum-Goldwasser asymmetric
key encryption protocol by using the proposed method.

The remainder of this paper  is organized as follows. In Section~\ref{section:related
  works} we  review some GPU implementations  of PRNGs.  Section~\ref{section:BASIC
  RECALLS} gives some basic recalls  on the well-known Devaney's formulation of chaos, 
  and on an iteration process called ``chaotic
iterations'' on which the post-treatment is based. 
The proposed PRNG and its proof of chaos are given in  Section~\ref{sec:pseudorandom}.
Section~\ref{sec:efficient    PRNG}   presents   an   efficient
implementation of  this chaotic PRNG  on a CPU, whereas   Section~\ref{sec:efficient PRNG
  gpu}   describes and evaluates theoretically  the  GPU   implementation. 
Such generators are experimented in 
Section~\ref{sec:experiments}.
We show in Section~\ref{sec:security analysis} that, if the inputted
generator is cryptographically secure, then it is the case too for the
generator provided by the post-treatment.
Such a proof leads to the proposition of a cryptographically secure and
chaotic generator on GPU based on the famous Blum Blum Shum
in Section~\ref{sec:CSGPU}, and to an improvement of the
Blum-Goldwasser protocol in Sect.~\ref{Blum-Goldwasser}.
This research work ends by a conclusion section, in which the contribution is
summarized and intended future work is presented.

\section{Related works on GPU based PRNGs}
\label{section:related works}

Numerous research works on defining GPU based PRNGs have already been proposed  in the
literature, so that exhaustivity is impossible.
This is why authors of this document only give reference to the most significant attempts 
in this domain, from their subjective point of view. 
The  quantity of pseudorandom numbers generated per second is mentioned here 
only when the information is given in the related work. 
A million numbers  per second will be simply written as
1MSample/s whereas a billion numbers per second is 1GSample/s.

In \cite{Pang:2008:cec}  a PRNG based on  cellular automata is defined
with no  requirement to an high  precision  integer   arithmetic  or to any bitwise
operations. Authors can   generate  about
3.2MSamples/s on a GeForce 7800 GTX GPU, which is quite an old card now.
However, there is neither a mention of statistical tests nor any proof of
chaos or cryptography in this document.

In \cite{ZRKB10}, the authors propose  different versions of efficient GPU PRNGs
based on  Lagged Fibonacci or Hybrid  Taus.  They have  used these
PRNGs   for  Langevin   simulations   of  biomolecules   fully  implemented   on
GPU. Performances of  the GPU versions are far better than  those obtained with a
CPU, and these PRNGs succeed to pass the {\it BigCrush} battery of TestU01. 
However the evaluations of the proposed PRNGs are only statistical ones.

Authors of~\cite{conf/fpga/ThomasHL09}  have studied the  implementation of some
PRNGs on  different computing architectures: CPU,  field-programmable gate array
(FPGA), massively parallel  processors, and GPU. This study is of interest, because
the  performance  of the  same  PRNGs on  different architectures are compared. 
FPGA appears as  the  fastest  and the most
efficient architecture, providing the fastest number of generated pseudorandom numbers
per joule. 
However, we notice that authors can ``only'' generate between 11 and 16GSamples/s
with a GTX 280  GPU, which should be compared with
the results presented in this document.
We can remark too that the PRNGs proposed in~\cite{conf/fpga/ThomasHL09} are only
able to pass the {\it Crush} battery, which is far easier than the {\it Big Crush} one.

Lastly, Cuda  has developed  a  library for  the  generation of  pseudorandom numbers  called
Curand~\cite{curand11}.        Several       PRNGs        are       implemented, among
other things 
Xorwow~\cite{Marsaglia2003} and  some variants of Sobol. The  tests reported show that
their  fastest version provides  15GSamples/s on  the new  Fermi C2050  card. 
But their PRNGs cannot pass the whole TestU01 battery (only one test is failed).
\newline
\newline
We can finally remark that, to the best of our knowledge, no GPU implementation has been proven to be chaotic, and the cryptographically secure property has surprisingly never been considered.

\section{Basic Recalls}
\label{section:BASIC RECALLS}

This section is devoted to basic definitions and terminologies in the fields of
topological chaos and chaotic iterations.
\subsection{Devaney's Chaotic Dynamical Systems}

In the sequel $S^{n}$ denotes the $n^{th}$ term of a sequence $S$ and $V_{i}$
denotes the $i^{th}$ component of a vector $V$. $f^{k}=f\circ ...\circ f$
is for the $k^{th}$ composition of a function $f$. Finally, the following
notation is used: $\llbracket1;N\rrbracket=\{1,2,\hdots,N\}$.

Consider a topological space $(\mathcal{X},\tau)$ and a continuous function $f :
\mathcal{X} \rightarrow \mathcal{X}$.

\begin{definition}
$f$ is said to be \emph{topologically transitive} if, for any pair of open sets
$U,V \subset \mathcal{X}$, there exists $k>0$ such that $f^k(U) \cap V \neq
\varnothing$.
\end{definition}

\begin{definition}
An element $x$ is a \emph{periodic point} for $f$ of period $n\in \mathds{N}^*$
if $f^{n}(x)=x$.
\end{definition}

\begin{definition}
$f$ is said to be \emph{regular} on $(\mathcal{X}, \tau)$ if the set of periodic
points for $f$ is dense in $\mathcal{X}$: for any point $x$ in $\mathcal{X}$,
any neighborhood of $x$ contains at least one periodic point (without
necessarily the same period).
\end{definition}

\begin{definition}[Devaney's formulation of chaos~\cite{Devaney}]
$f$ is said to be \emph{chaotic} on $(\mathcal{X},\tau)$ if $f$ is regular and
topologically transitive.
\end{definition}

The chaos property is strongly linked to the notion of ``sensitivity'', defined
on a metric space $(\mathcal{X},d)$ by:

\begin{definition}
\label{sensitivity} $f$ has \emph{sensitive dependence on initial conditions}
if there exists $\delta >0$ such that, for any $x\in \mathcal{X}$ and any
neighborhood $V$ of $x$, there exist $y\in V$ and $n > 0$ such that
$d\left(f^{n}(x), f^{n}(y)\right) >\delta $.

$\delta$ is called the \emph{constant of sensitivity} of $f$.
\end{definition}

Indeed, Banks \emph{et al.} have proven in~\cite{Banks92} that when $f$ is
chaotic and $(\mathcal{X}, d)$ is a metric space, then $f$ has the property of
sensitive dependence on initial conditions (this property was formerly an
element of the definition of chaos). To sum up, quoting Devaney
in~\cite{Devaney}, a chaotic dynamical system ``is unpredictable because of the
sensitive dependence on initial conditions. It cannot be broken down or
simplified into two subsystems which do not interact because of topological
transitivity. And in the midst of this random behavior, we nevertheless have an
element of regularity''. Fundamentally different behaviors are consequently
possible and occur in an unpredictable way.

\subsection{Chaotic Iterations}
\label{sec:chaotic iterations}

Let us consider  a \emph{system} with a finite  number $\mathsf{N} \in
\mathds{N}^*$ of elements  (or \emph{cells}), so that each  cell has a
Boolean  \emph{state}. Having $\mathsf{N}$ Boolean values for these
 cells  leads to the definition of a particular \emph{state  of the
system}. A sequence which  elements belong to $\llbracket 1;\mathsf{N}
\rrbracket $ is called a \emph{strategy}. The set of all strategies is
denoted by $\llbracket 1, \mathsf{N} \rrbracket^\mathds{N}.$

\begin{definition}
\label{Def:chaotic iterations}
The      set       $\mathds{B}$      denoting      $\{0,1\}$,      let
$f:\mathds{B}^{\mathsf{N}}\longrightarrow  \mathds{B}^{\mathsf{N}}$ be
a  function  and  $S\in  \llbracket 1, \mathsf{N} \rrbracket^\mathds{N}$  be  a  ``strategy''.  The  so-called
\emph{chaotic      iterations}     are     defined      by     $x^0\in
\mathds{B}^{\mathsf{N}}$ and
\begin{equation}
\forall    n\in     \mathds{N}^{\ast     },    \forall     i\in
\llbracket1;\mathsf{N}\rrbracket ,x_i^n=\left\{
\begin{array}{ll}
  x_i^{n-1} &  \text{ if  }S^n\neq i \\
  \left(f(x^{n-1})\right)_{S^n} & \text{ if }S^n=i.
\end{array}\right.
\end{equation}
\end{definition}

In other words, at the $n^{th}$ iteration, only the $S^{n}-$th cell is
\textquotedblleft  iterated\textquotedblright .  Note  that in  a more
general  formulation,  $S^n$  can   be  a  subset  of  components  and
$\left(f(x^{n-1})\right)_{S^{n}}$      can     be      replaced     by
$\left(f(x^{k})\right)_{S^{n}}$, where  $k<n$, describing for example,
delays  transmission~\cite{Robert1986,guyeux10}.  Finally,  let us  remark that
the term  ``chaotic'', in  the name of  these iterations,  has \emph{a
priori} no link with the mathematical theory of chaos, presented above.

Let us now recall how to define a suitable metric space where chaotic iterations
are continuous. For further explanations, see, e.g., \cite{guyeux10}.

Let $\delta $ be the \emph{discrete Boolean metric}, $\delta
(x,y)=0\Leftrightarrow x=y.$ Given a function $f$, define the function:
\begin{equation}
\begin{array}{lrll}
F_{f}: & \llbracket1;\mathsf{N}\rrbracket\times \mathds{B}^{\mathsf{N}} &
\longrightarrow & \mathds{B}^{\mathsf{N}} \\
& (k,E) & \longmapsto & \left( E_{j}.\delta (k,j)+f(E)_{k}.\overline{\delta
(k,j)}\right) _{j\in \llbracket1;\mathsf{N}\rrbracket},%
\end{array}%
\end{equation}%
\noindent where + and . are the Boolean addition and product operations.
Consider the phase space:
\begin{equation}
\mathcal{X} = \llbracket 1 ; \mathsf{N} \rrbracket^\mathds{N} \times
\mathds{B}^\mathsf{N},
\end{equation}
\noindent and the map defined on $\mathcal{X}$:
\begin{equation}
G_f\left(S,E\right) = \left(\sigma(S), F_f(i(S),E)\right), \label{Gf}
\end{equation}
\noindent where $\sigma$ is the \emph{shift} function defined by $\sigma
(S^{n})_{n\in \mathds{N}}\in \llbracket 1, \mathsf{N} \rrbracket^\mathds{N}\longrightarrow (S^{n+1})_{n\in
\mathds{N}}\in \llbracket 1, \mathsf{N} \rrbracket^\mathds{N}$ and $i$ is the \emph{initial function} 
$i:(S^{n})_{n\in \mathds{N}} \in \llbracket 1, \mathsf{N} \rrbracket^\mathds{N}\longrightarrow S^{0}\in \llbracket
1;\mathsf{N}\rrbracket$. Then the chaotic iterations proposed in
Definition \ref{Def:chaotic iterations} can be described by the following iterations:
\begin{equation}
\left\{
\begin{array}{l}
X^0 \in \mathcal{X} \\
X^{k+1}=G_{f}(X^k).%
\end{array}%
\right.
\end{equation}%

With this formulation, a shift function appears as a component of chaotic
iterations. The shift function is a famous example of a chaotic
map~\cite{Devaney} but its presence is not sufficient enough to claim $G_f$ as
chaotic. 
To study this claim, a new distance between two points $X = (S,E), Y =
(\check{S},\check{E})\in
\mathcal{X}$ has been introduced in \cite{guyeux10} as follows:
\begin{equation}
d(X,Y)=d_{e}(E,\check{E})+d_{s}(S,\check{S}),
\end{equation}
\noindent where
\begin{equation}
\left\{
\begin{array}{lll}
\displaystyle{d_{e}(E,\check{E})} & = & \displaystyle{\sum_{k=1}^{\mathsf{N}%
}\delta (E_{k},\check{E}_{k})}, \\
\displaystyle{d_{s}(S,\check{S})} & = & \displaystyle{\dfrac{9}{\mathsf{N}}%
\sum_{k=1}^{\infty }\dfrac{|S^k-\check{S}^k|}{10^{k}}}.%
\end{array}%
\right.
\end{equation}

This new distance has been introduced to satisfy the following requirements.
\begin{itemize}
\item When the number of different cells between two systems is increasing, then
their distance should increase too.
\item In addition, if two systems present the same cells and their respective
strategies start with the same terms, then the distance between these two points
must be small because the evolution of the two systems will be the same for a
while. Indeed, both dynamical systems start with the same initial condition,
use the same update function, and as strategies are the same for a while, furthermore
updated components are the same as well.
\end{itemize}
The distance presented above follows these recommendations. Indeed, if the floor
value $\lfloor d(X,Y)\rfloor $ is equal to $n$, then the systems $E, \check{E}$
differ in $n$ cells ($d_e$ is indeed the Hamming distance). In addition, $d(X,Y) - \lfloor d(X,Y) \rfloor $ is a
measure of the differences between strategies $S$ and $\check{S}$. More
precisely, this floating part is less than $10^{-k}$ if and only if the first
$k$ terms of the two strategies are equal. Moreover, if the $k^{th}$ digit is
nonzero, then the $k^{th}$ terms of the two strategies are different.
The impact of this choice for a distance will be investigated at the end of the document.

Finally, it has been established in \cite{guyeux10} that,

\begin{proposition}
Let $f$ be a map from $\mathds{B}^\mathsf{N}$ to itself. Then $G_{f}$ is continuous in
the metric space $(\mathcal{X},d)$.
\end{proposition}

The chaotic property of $G_f$ has been firstly established for the vectorial
Boolean negation $f(x_1,\hdots, x_\mathsf{N}) =  (\overline{x_1},\hdots, \overline{x_\mathsf{N}})$ \cite{guyeux10}. To obtain a characterization, we have secondly
introduced the notion of asynchronous iteration graph recalled bellow.

Let $f$ be a map from $\mathds{B}^\mathsf{N}$ to itself. The
{\emph{asynchronous iteration graph}} associated with $f$ is the
directed graph $\Gamma(f)$ defined by: the set of vertices is
$\mathds{B}^\mathsf{N}$; for all $x\in\mathds{B}^\mathsf{N}$ and 
$i\in \llbracket1;\mathsf{N}\rrbracket$,
the graph $\Gamma(f)$ contains an arc from $x$ to $F_f(i,x)$. 
The relation between $\Gamma(f)$ and $G_f$ is clear: there exists a
path from $x$ to $x'$ in $\Gamma(f)$ if and only if there exists a
strategy $s$ such that the parallel iteration of $G_f$ from the
initial point $(s,x)$ reaches the point $x'$.
We have then proven in \cite{bcgr11:ip} that,

\begin{theorem}
\label{Th:Caractérisation   des   IC   chaotiques}  
Let $f:\mathds{B}^\mathsf{N}\to\mathds{B}^\mathsf{N}$. $G_f$ is chaotic  (according to  Devaney) 
if and only if $\Gamma(f)$ is strongly connected.
\end{theorem}

Finally, we have established in \cite{bcgr11:ip} that,
\begin{theorem}
  Let $f: \mathds{B}^{n} \rightarrow \mathds{B}^{n}$, $\Gamma(f)$ its
  iteration graph, $\check{M}$ its adjacency
  matrix and $M$
  a $n\times n$ matrix defined by 
  $
  M_{ij} = \frac{1}{n}\check{M}_{ij}$ 
  if $i \neq j$ and  
  $M_{ii} = 1 - \frac{1}{n} \sum\limits_{j=1, j\neq i}^n \check{M}_{ij}$ otherwise.
  
  If $\Gamma(f)$ is strongly connected, then 
  the output of the PRNG detailed in Algorithm~\ref{CI Algorithm} follows 
  a law that tends to the uniform distribution 
  if and only if $M$ is a double stochastic matrix.
\end{theorem}

These results of chaos and uniform distribution have led us to study the possibility of building a
pseudorandom number generator (PRNG) based on the chaotic iterations. 
As $G_f$, defined on the domain   $\llbracket 1 ;  \mathsf{N} \rrbracket^{\mathds{N}} 
\times \mathds{B}^\mathsf{N}$, is built from Boolean networks $f : \mathds{B}^\mathsf{N}
\rightarrow \mathds{B}^\mathsf{N}$, we can preserve the theoretical properties on $G_f$
during implementations (due to the discrete nature of $f$). Indeed, it is as if
$\mathds{B}^\mathsf{N}$ represents the memory of the computer whereas $\llbracket 1 ;  \mathsf{N}
\rrbracket^{\mathds{N}}$ is its input stream (the seeds, for instance, in PRNG, or a physical noise in TRNG).
Let us finally remark that the vectorial negation satisfies the hypotheses of both theorems above.

\section{Application to Pseudorandomness}
\label{sec:pseudorandom}

\subsection{A First Pseudorandom Number Generator}

We have proposed in~\cite{bgw09:ip} a new family of generators that receives 
two PRNGs as inputs. These two generators are mixed with chaotic iterations, 
leading thus to a new PRNG that improves the statistical properties of each
generator taken alone. Furthermore, our generator 
possesses various chaos properties that none of the generators used as input
present.

\begin{algorithm}[h!]
\KwIn{a function $f$, an iteration number $b$, an initial configuration $x^0$
($n$ bits)}
\KwOut{a configuration $x$ ($n$ bits)}
$x\leftarrow x^0$\;
$k\leftarrow b + \textit{XORshift}(b)$\;
\For{$i=0,\dots,k$}
{
$s\leftarrow{\textit{XORshift}(n)}$\;
$x\leftarrow{F_f(s,x)}$\;
}
return $x$\;
\caption{PRNG with chaotic functions}
\label{CI Algorithm}
\end{algorithm}

\begin{algorithm}[h!]
\KwIn{the internal configuration $z$ (a 32-bit word)}
\KwOut{$y$ (a 32-bit word)}
$z\leftarrow{z\oplus{(z\ll13)}}$\;
$z\leftarrow{z\oplus{(z\gg17)}}$\;
$z\leftarrow{z\oplus{(z\ll5)}}$\;
$y\leftarrow{z}$\;
return $y$\;
\medskip
\caption{An arbitrary round of \textit{XORshift} algorithm}
\label{XORshift}
\end{algorithm}

This generator is synthesized in Algorithm~\ref{CI Algorithm}.
It takes as input: a Boolean function $f$ satisfying Theorem~\ref{Th:Caractérisation   des   IC   chaotiques};
an integer $b$, ensuring that the number of executed iterations is at least $b$
and at most $2b+1$; and an initial configuration $x^0$.
It returns the new generated configuration $x$.  Internally, it embeds two
\textit{XORshift}$(k)$ PRNGs~\cite{Marsaglia2003} that return integers
uniformly distributed
into $\llbracket 1 ; k \rrbracket$.
\textit{XORshift} is a category of very fast PRNGs designed by George Marsaglia,
which repeatedly uses the transform of exclusive or (XOR, $\oplus$) on a number
with a bit shifted version of it. This PRNG, which has a period of
$2^{32}-1=4.29\times10^9$, is summed up in Algorithm~\ref{XORshift}. It is used
in our PRNG to compute the strategy length and the strategy elements.

This former generator has successively passed various batteries of statistical tests, as the NIST~\cite{bcgr11:ip}, DieHARD~\cite{Marsaglia1996}, and TestU01~\cite{LEcuyerS07} ones.

\subsection{Improving the Speed of the Former Generator}

Instead of updating only one cell at each iteration, we can try to choose a
subset of components and to update them together. Such an attempt leads
to a kind of merger of the two sequences used in Algorithm 
\ref{CI Algorithm}. When the updating function is the vectorial negation,
this algorithm can be rewritten as follows:

\begin{equation}
\left\{
\begin{array}{l}
x^0 \in \llbracket 0, 2^\mathsf{N}-1 \rrbracket, S \in \llbracket 0, 2^\mathsf{N}-1 \rrbracket^\mathds{N} \\
\forall n \in \mathds{N}^*, x^n = x^{n-1} \oplus S^n,
\end{array}
\right.
\label{equation Oplus}
\end{equation}
where $\oplus$ is for the bitwise exclusive or between two integers. 
This rewriting can be understood as follows. The $n-$th term $S^n$ of the
sequence $S$, which is an integer of $\mathsf{N}$ binary digits, presents
the list of cells to update in the state $x^n$ of the system (represented
as an integer having $\mathsf{N}$ bits too). More precisely, the $k-$th 
component of this state (a binary digit) changes if and only if the $k-$th 
digit in the binary decomposition of $S^n$ is 1.

The single basic component presented in Eq.~\ref{equation Oplus} is of 
ordinary use as a good elementary brick in various PRNGs. It corresponds
to the following discrete dynamical system in chaotic iterations:

\begin{equation}
\forall    n\in     \mathds{N}^{\ast     },    \forall     i\in
\llbracket1;\mathsf{N}\rrbracket ,x_i^n=\left\{
\begin{array}{ll}
  x_i^{n-1} &  \text{ if  } i \notin \mathcal{S}^n \\
  \left(f(x^{n-1})\right)_{S^n} & \text{ if }i \in \mathcal{S}^n.
\end{array}\right.
\label{eq:generalIC}
\end{equation}
where $f$ is the vectorial negation and $\forall n \in \mathds{N}$, 
$\mathcal{S}^n \subset \llbracket 1, \mathsf{N} \rrbracket$ is such that
$k \in \mathcal{S}^n$ if and only if the $k-$th digit in the binary
decomposition of $S^n$ is 1. Such chaotic iterations are more general
than the ones presented in Definition \ref{Def:chaotic iterations} because, instead of updating only one term at each iteration,
we select a subset of components to change.

Obviously, replacing Algorithm~\ref{CI Algorithm} by 
Equation~\ref{equation Oplus}, which is possible when the iteration function is
the vectorial negation, leads to a speed improvement. However, proofs
of chaos obtained in~\cite{bg10:ij} have been established
only for chaotic iterations of the form presented in Definition 
\ref{Def:chaotic iterations}. The question is now to determine whether the
use of more general chaotic iterations to generate pseudorandom numbers 
faster, does not deflate their topological chaos properties.

\subsection{Proofs of Chaos of the General Formulation of the Chaotic Iterations}
\label{deuxième def}
Let us consider the discrete dynamical systems in chaotic iterations having 
the general form:

\begin{equation}
\forall    n\in     \mathds{N}^{\ast     },    \forall     i\in
\llbracket1;\mathsf{N}\rrbracket ,x_i^n=\left\{
\begin{array}{ll}
  x_i^{n-1} &  \text{ if  } i \notin \mathcal{S}^n \\
  \left(f(x^{n-1})\right)_{S^n} & \text{ if }i \in \mathcal{S}^n.
\end{array}\right.
\label{general CIs}
\end{equation}

In other words, at the $n^{th}$ iteration, only the cells whose id is
contained into the set $S^{n}$ are iterated.

Let us now rewrite these general chaotic iterations as usual discrete dynamical
system of the form $X^{n+1}=f(X^n)$ on an ad hoc metric space. Such a formulation
is required in order to study the topological behavior of the system.

Let us introduce the following function:
\begin{equation}
\begin{array}{cccc}
 \chi: & \llbracket 1; \mathsf{N} \rrbracket \times \mathcal{P}\left(\llbracket 1; \mathsf{N} \rrbracket\right) & \longrightarrow & \mathds{B}\\
         & (i,X) & \longmapsto  & \left\{ \begin{array}{ll} 0 & \textrm{if }i \notin X, \\ 1 & \textrm{if }i \in X,  \end{array}\right.
\end{array} 
\end{equation}
where $\mathcal{P}\left(X\right)$ is for the powerset of the set $X$, that is, $Y \in \mathcal{P}\left(X\right) \Longleftrightarrow Y \subset X$.

Given a function $f:\mathds{B}^\mathsf{N} \longrightarrow \mathds{B}^\mathsf{N} $, define the function:
\begin{equation}
\begin{array}{lrll}
F_{f}: & \mathcal{P}\left(\llbracket1;\mathsf{N}\rrbracket \right) \times \mathds{B}^{\mathsf{N}} &
\longrightarrow & \mathds{B}^{\mathsf{N}} \\
& (P,E) & \longmapsto & \left( E_{j}.\chi (j,P)+f(E)_{j}.\overline{\chi
(j,P)}\right) _{j\in \llbracket1;\mathsf{N}\rrbracket},%
\end{array}%
\end{equation}%
where + and . are the Boolean addition and product operations, and $\overline{x}$ 
is the negation of the Boolean $x$.
Consider the phase space:
\begin{equation}
\mathcal{X} = \mathcal{P}\left(\llbracket 1 ; \mathsf{N} \rrbracket\right)^\mathds{N} \times
\mathds{B}^\mathsf{N},
\end{equation}
\noindent and the map defined on $\mathcal{X}$:
\begin{equation}
G_f\left(S,E\right) = \left(\sigma(S), F_f(i(S),E)\right), \label{Gf}
\end{equation}
\noindent where $\sigma$ is the \emph{shift} function defined by $\sigma
(S^{n})_{n\in \mathds{N}}\in \mathcal{P}\left(\llbracket 1 ; \mathsf{N} \rrbracket\right)^\mathds{N}\longrightarrow (S^{n+1})_{n\in
\mathds{N}}\in \mathcal{P}\left(\llbracket 1 ; \mathsf{N} \rrbracket\right)^\mathds{N}$ and $i$ is the \emph{initial function} 
$i:(S^{n})_{n\in \mathds{N}} \in \mathcal{P}\left(\llbracket 1 ; \mathsf{N} \rrbracket\right)^\mathds{N}\longrightarrow S^{0}\in \mathcal{P}\left(\llbracket 1 ; \mathsf{N} \rrbracket\right)$. 
Then the general chaotic iterations defined in Equation \ref{general CIs} can 
be described by the following discrete dynamical system:
\begin{equation}
\left\{
\begin{array}{l}
X^0 \in \mathcal{X} \\
X^{k+1}=G_{f}(X^k).%
\end{array}%
\right.
\end{equation}%

Once more, a shift function appears as a component of these general chaotic 
iterations. 

To study the Devaney's chaos property, a distance between two points 
$X = (S,E), Y = (\check{S},\check{E})$ of $\mathcal{X}$ must be defined.
Let us introduce:
\begin{equation}
d(X,Y)=d_{e}(E,\check{E})+d_{s}(S,\check{S}),
\label{nouveau d}
\end{equation}
\noindent where
\begin{equation}
\left\{
\begin{array}{lll}
\displaystyle{d_{e}(E,\check{E})} & = & \displaystyle{\sum_{k=1}^{\mathsf{N}%
}\delta (E_{k},\check{E}_{k})}\textrm{ is once more the Hamming distance}, \\
\displaystyle{d_{s}(S,\check{S})} & = & \displaystyle{\dfrac{9}{\mathsf{N}}%
\sum_{k=1}^{\infty }\dfrac{|S^k\Delta {S}^k|}{10^{k}}}.%
\end{array}%
\right.
\end{equation}
where $|X|$ is the cardinality of a set $X$ and $A\Delta B$ is for the symmetric difference, defined for sets A, B as
$A\,\Delta\,B = (A \setminus B) \cup (B \setminus A)$.

\begin{proposition}
The function $d$ defined in Eq.~\ref{nouveau d} is a metric on $\mathcal{X}$.
\end{proposition}

\begin{proof}
 $d_e$ is the Hamming distance. We will prove that $d_s$ is a distance
too, thus $d$, as being the sum of two distances, will also be a distance.
 \begin{itemize}
\item Obviously, $d_s(S,\check{S})\geqslant 0$, and if $S=\check{S}$, then 
$d_s(S,\check{S})=0$. Conversely, if $d_s(S,\check{S})=0$, then 
$\forall k \in \mathds{N}, |S^k\Delta {S}^k|=0$, and so $\forall k, S^k=\check{S}^k$.
 \item $d_s$ is symmetric 
($d_s(S,\check{S})=d_s(\check{S},S)$) due to the commutative property
of the symmetric difference. 
\item Finally, $|S \Delta S''| = |(S \Delta \varnothing) \Delta S''|= |S \Delta (S'\Delta S') \Delta S''|= |(S \Delta S') \Delta (S' \Delta S'')|\leqslant |S \Delta S'| + |S' \Delta S''|$, 
and so for all subsets $S,S',$ and $S''$ of $\llbracket 1, \mathsf{N} \rrbracket$, 
we have $d_s(S,S'') \leqslant d_e(S,S')+d_s(S',S'')$, and the triangle
inequality is obtained.
 \end{itemize}
\end{proof}

Before being able to study the topological behavior of the general 
chaotic iterations, we must first establish that:

\begin{proposition}
 For all $f:\mathds{B}^\mathsf{N} \longrightarrow \mathds{B}^\mathsf{N} $, the function $G_f$ is continuous on 
$\left( \mathcal{X},d\right)$.
\end{proposition}

\begin{proof}
We use the sequential continuity.
Let $(S^n,E^n)_{n\in \mathds{N}}$ be a sequence of the phase space $%
\mathcal{X}$, which converges to $(S,E)$. We will prove that $\left(
G_{f}(S^n,E^n)\right) _{n\in \mathds{N}}$ converges to $\left(
G_{f}(S,E)\right) $. Let us remark that for all $n$, $S^n$ is a strategy,
thus, we consider a sequence of strategies (\emph{i.e.}, a sequence of
sequences).\newline
As $d((S^n,E^n);(S,E))$ converges to 0, each distance $d_{e}(E^n,E)$ and $d_{s}(S^n,S)$ converges
to 0. But $d_{e}(E^n,E)$ is an integer, so $\exists n_{0}\in \mathds{N},$ $%
d_{e}(E^n,E)=0$ for any $n\geqslant n_{0}$.\newline
In other words, there exists a threshold $n_{0}\in \mathds{N}$ after which no
cell will change its state:
$\exists n_{0}\in \mathds{N},n\geqslant n_{0}\Rightarrow E^n = E.$

In addition, $d_{s}(S^n,S)\longrightarrow 0,$ so $\exists n_{1}\in %
\mathds{N},d_{s}(S^n,S)<10^{-1}$ for all indexes greater than or equal to $%
n_{1}$. This means that for $n\geqslant n_{1}$, all the $S^n$ have the same
first term, which is $S^0$: $\forall n\geqslant n_{1},S_0^n=S_0.$

Thus, after the $max(n_{0},n_{1})^{th}$ term, states of $E^n$ and $E$ are
identical and strategies $S^n$ and $S$ start with the same first term.\newline
Consequently, states of $G_{f}(S^n,E^n)$ and $G_{f}(S,E)$ are equal,
so, after the $max(n_0, n_1)^{th}$ term, the distance $d$ between these two points is strictly less than 1.\newline
\noindent We now prove that the distance between $\left(
G_{f}(S^n,E^n)\right) $ and $\left( G_{f}(S,E)\right) $ is convergent to
0. Let $\varepsilon >0$. \medskip
\begin{itemize}
\item If $\varepsilon \geqslant 1$, we see that the distance
between $\left( G_{f}(S^n,E^n)\right) $ and $\left( G_{f}(S,E)\right) $ is
strictly less than 1 after the $max(n_{0},n_{1})^{th}$ term (same state).
\medskip
\item If $\varepsilon <1$, then $\exists k\in \mathds{N},10^{-k}\geqslant
\varepsilon > 10^{-(k+1)}$. But $d_{s}(S^n,S)$ converges to 0, so
\begin{equation*}
\exists n_{2}\in \mathds{N},\forall n\geqslant
n_{2},d_{s}(S^n,S)<10^{-(k+2)},
\end{equation*}%
thus after $n_{2}$, the $k+2$ first terms of $S^n$ and $S$ are equal.
\end{itemize}
\noindent As a consequence, the $k+1$ first entries of the strategies of $%
G_{f}(S^n,E^n)$ and $G_{f}(S,E)$ are the same ($G_{f}$ is a shift of strategies) and due to the definition of $d_{s}$, the floating part of
the distance between $(S^n,E^n)$ and $(S,E)$ is strictly less than $%
10^{-(k+1)}\leqslant \varepsilon $.\bigskip \newline
In conclusion,
$$
\forall \varepsilon >0,\exists N_{0}=max(n_{0},n_{1},n_{2})\in \mathds{N}%
,\forall n\geqslant N_{0},
 d\left( G_{f}(S^n,E^n);G_{f}(S,E)\right)
\leqslant \varepsilon .
$$
$G_{f}$ is consequently continuous.
\end{proof}

It is now possible to study the topological behavior of the general chaotic
iterations. We will prove that,

\begin{theorem}
\label{t:chaos des general}
 The general chaotic iterations defined on Equation~\ref{general CIs} satisfy
the Devaney's property of chaos.
\end{theorem}

Let us firstly prove the following lemma.

\begin{lemma}[Strong transitivity]
\label{strongTrans}
 For all couples $X,Y \in \mathcal{X}$ and any neighborhood $V$ of $X$, we can 
find $n \in \mathds{N}^*$ and $X' \in V$ such that $G^n(X')=Y$.
\end{lemma}

\begin{proof}
 Let $X=(S,E)$, $\varepsilon>0$, and $k_0 = \lfloor log_{10}(\varepsilon)+1 \rfloor$. 
Any point $X'=(S',E')$ such that $E'=E$ and $\forall k \leqslant k_0, S'^k=S^k$, 
are in the open ball $\mathcal{B}\left(X,\varepsilon\right)$. Let us define 
$\check{X} = \left(\check{S},\check{E}\right)$, where $\check{X}= G^{k_0}(X)$.
We denote by $s\subset \llbracket 1; \mathsf{N} \rrbracket$ the set of coordinates
that are different between $\check{E}$ and the state of $Y$. Thus each point $X'$ of
the form $(S',E')$ where $E'=E$ and $S'$ starts with 
$(S^0, S^1, \hdots, S^{k_0},s,\hdots)$, verifies the following properties:
\begin{itemize}
 \item $X'$ is in $\mathcal{B}\left(X,\varepsilon\right)$,
 \item the state of $G_f^{k_0+1}(X')$ is the state of $Y$.
\end{itemize}
Finally the point $\left(\left(S^0, S^1, \hdots, S^{k_0},s,s^0, s^1, \hdots\right); E\right)$, 
where $(s^0,s^1, \hdots)$ is the strategy of $Y$, satisfies the properties
claimed in the lemma.
\end{proof}

We can now prove Theorem~\ref{t:chaos des general}...

\begin{proof}[Theorem~\ref{t:chaos des general}]
Firstly, strong transitivity implies transitivity.

Let $(S,E) \in\mathcal{X}$ and $\varepsilon >0$. To
prove that $G_f$ is regular, it is sufficient to prove that
there exists a strategy $\tilde S$ such that the distance between
$(\tilde S,E)$ and $(S,E)$ is less than $\varepsilon$, and such that
$(\tilde S,E)$ is a periodic point.

Let $t_1=\lfloor-\log_{10}(\varepsilon)\rfloor$, and let $E'$ be the
configuration that we obtain from $(S,E)$ after $t_1$ iterations of
$G_f$. As $G_f$ is strongly transitive, there exists a strategy $S'$ 
and $t_2\in\mathds{N}$ such
that $E$ is reached from $(S',E')$ after $t_2$ iterations of $G_f$.

Consider the strategy $\tilde S$ that alternates the first $t_1$ terms
of $S$ and the first $t_2$ terms of $S'$: $$\tilde
S=(S_0,\dots,S_{t_1-1},S'_0,\dots,S'_{t_2-1},S_0,\dots,S_{t_1-1},S'_0,\dots,S'_{t_2-1},S_0,\dots).$$ It
is clear that $(\tilde S,E)$ is obtained from $(\tilde S,E)$ after
$t_1+t_2$ iterations of $G_f$. So $(\tilde S,E)$ is a periodic
point. Since $\tilde S_t=S_t$ for $t<t_1$, by the choice of $t_1$, we
have $d((S,E),(\tilde S,E))<\epsilon$.
\end{proof}

\section{Efficient PRNG based on Chaotic Iterations}
\label{sec:efficient PRNG}

Based on the proof presented in the previous section, it is now possible to 
improve the speed of the generator formerly presented in~\cite{bgw09:ip,guyeux10}. 
The first idea is to consider
that the provided strategy is a pseudorandom Boolean vector obtained by a
given PRNG.
An iteration of the system is simply the bitwise exclusive or between
the last computed state and the current strategy.
Topological properties of disorder exhibited by chaotic 
iterations can be inherited by the inputted generator, we hope by doing so to 
obtain some statistical improvements while preserving speed.

Let us give an example using 16-bits numbers, to clearly understand how the bitwise xor operations
are
done.  
Suppose  that $x$ and the  strategy $S^i$ are given as
binary vectors.
Table~\ref{TableExemple} shows the result of $x \oplus S^i$.

\begin{table}
$$
\begin{array}{|cc|cccccccccccccccc|}
\hline
x      &=&1&0&1&1&1&0&1&0&1&0&0&1&0&0&1&0\\
\hline
S^i      &=&0&1&1&0&0&1&1&0&1&1&1&0&0&1&1&1\\
\hline
x \oplus S^i&=&1&1&0&1&1&1&0&0&0&1&1&1&0&1&0&1\\
\hline

\hline
 \end{array}
$$
\caption{Example of an arbitrary round of the proposed generator}
\label{TableExemple}
\end{table}

\lstset{language=C,caption={C code of the sequential PRNG based on chaotic iteration\
s},label=algo:seqCIPRNG}
\begin{lstlisting}
unsigned int CIPRNG() {
  static unsigned int x = 123123123;
  unsigned long t1 = xorshift();
  unsigned long t2 = xor128();
  unsigned long t3 = xorwow();
  x = x^(unsigned int)t1;
  x = x^(unsigned int)(t2>>32);
  x = x^(unsigned int)(t3>>32);
  x = x^(unsigned int)t2;
  x = x^(unsigned int)(t1>>32);
  x = x^(unsigned int)t3;
  return x;
}
\end{lstlisting}

In Listing~\ref{algo:seqCIPRNG} a sequential  version of the proposed PRNG based
on  chaotic  iterations  is  presented.   The xor  operator  is  represented  by
\textasciicircum.  This function uses  three classical 64-bits PRNGs, namely the
\texttt{xorshift},         the          \texttt{xor128},         and         the
\texttt{xorwow}~\cite{Marsaglia2003}.  In the following, we call them ``xor-like
PRNGs''.   As each  xor-like PRNG  uses 64-bits  whereas our  proposed generator
works with 32-bits, we use the command \texttt{(unsigned int)}, that selects the
32 least  significant bits  of a given  integer, and the  code \texttt{(unsigned
  int)(t$>>$32)} in order to obtain the 32 most significant bits of \texttt{t}.

Thus producing a pseudorandom number needs 6 xor operations with 6 32-bits numbers
that  are provided by  3 64-bits  PRNGs.  This  version successfully  passes the
stringent BigCrush battery of tests~\cite{LEcuyerS07}.

\section{Efficient PRNGs based on Chaotic Iterations on GPU}
\label{sec:efficient PRNG gpu}

In order to  take benefits from the computing power  of GPU, a program
needs  to have  independent blocks  of  threads that  can be  computed
simultaneously. In general,  the larger the number of  threads is, the
more local  memory is  used, and the  less branching  instructions are
used  (if,  while,  ...),  the  better the  performances  on  GPU  is.
Obviously, having these requirements in  mind, it is possible to build
a   program    similar   to    the   one   presented    in  Listing 
\ref{algo:seqCIPRNG}, which computes  pseudorandom numbers on GPU.  To
do  so,  we  must   firstly  recall  that  in  the  CUDA~\cite{Nvid10}
environment,    threads    have     a    local    identifier    called
\texttt{ThreadIdx},  which   is  relative  to   the  block  containing
them. Furthermore, in  CUDA, parts of  the code that are executed by the  GPU, are
called {\it kernels}.

\subsection{Naive Version for GPU}

It is possible to deduce from the CPU version a quite similar version adapted to GPU.
The simple principle consists in making each thread of the GPU computing the CPU version of our PRNG.  
Of course,  the  three xor-like
PRNGs  used in these computations must have different  parameters. 
In a given thread, these parameters are
randomly picked from another PRNGs. 
The  initialization stage is performed by  the CPU.
To do it, the  ISAAC  PRNG~\cite{Jenkins96} is used to  set  all  the
parameters embedded into each thread.   

The implementation of  the three
xor-like  PRNGs  is  straightforward  when  their  parameters  have  been
allocated in  the GPU memory.  Each xor-like  works with  an internal
number  $x$  that saves  the  last  generated  pseudorandom number. Additionally,  the
implementation of the  xor128, the xorshift, and the  xorwow respectively require
4, 5, and 6 unsigned long as internal variables.

\begin{algorithm}

\KwIn{InternalVarXorLikeArray: array with internal variables of the 3 xor-like
PRNGs in global memory\;
NumThreads: number of threads\;}
\KwOut{NewNb: array containing random numbers in global memory}
\If{threadIdx is concerned by the computation} {
  retrieve data from InternalVarXorLikeArray[threadIdx] in local variables\;
  \For{i=1 to n} {
    compute a new PRNG as in Listing\ref{algo:seqCIPRNG}\;
    store the new PRNG in NewNb[NumThreads*threadIdx+i]\;
  }
  store internal variables in InternalVarXorLikeArray[threadIdx]\;
}

\caption{Main kernel of the GPU ``naive'' version of the PRNG based on chaotic iterations}
\label{algo:gpu_kernel}
\end{algorithm}

Algorithm~\ref{algo:gpu_kernel}  presents a naive  implementation of the proposed  PRNG on
GPU.  Due to the available  memory in the  GPU and the number  of threads
used simultaneously,  the number  of random numbers  that a thread  can generate
inside   a    kernel   is   limited  (\emph{i.e.},    the    variable   \texttt{n}   in
algorithm~\ref{algo:gpu_kernel}). For instance, if  $100,000$ threads are used and
if $n=100$\footnote{in fact, we need to add the initial seed (a 32-bits number)},
then   the  memory   required   to  store all of the  internals   variables  of both the  xor-like
PRNGs\footnote{we multiply this number by $2$ in order to count 32-bits numbers}
and  the pseudorandom  numbers generated by  our  PRNG,  is  equal to  $100,000\times  ((4+5+6)\times
2+(1+100))=1,310,000$ 32-bits numbers, that is, approximately $52$Mb.

This generator is able to pass the whole BigCrush battery of tests, for all
the versions that have been tested depending on their number of threads 
(called \texttt{NumThreads} in our algorithm, tested up to $5$ million).

\begin{remark}
The proposed algorithm has  the  advantage of  manipulating  independent
PRNGs, so this version is easily adaptable on a cluster of computers too. The only thing
to ensure is to use a single ISAAC PRNG. To achieve this requirement, a simple solution consists in
using a master node for the initialization. This master node computes the initial parameters
for all the different nodes involved in the computation.
\end{remark}

\subsection{Improved Version for GPU}

As GPU cards using CUDA have shared memory between threads of the same block, it
is possible  to use this  feature in order  to simplify the  previous algorithm,
i.e., to use less  than 3 xor-like PRNGs. The solution  consists in computing only
one xor-like PRNG by thread, saving  it into the shared memory, and then to use the results
of some  other threads in the  same block of  threads. In order to  define which
thread uses the result of which other  one, we can use a combination array that
contains  the indexes  of  all threads  and  for which  a combination has  been
performed. 

In  Algorithm~\ref{algo:gpu_kernel2},  two  combination  arrays are  used.   The
variable     \texttt{offset}    is     computed    using     the     value    of
\texttt{combination\_size}.   Then we  can compute  \texttt{o1}  and \texttt{o2}
representing the  indexes of  the other  threads whose results  are used  by the
current one.   In this algorithm, we  consider that a 32-bits  xor-like PRNG has
been chosen. In practice, we  use the xor128 proposed in~\cite{Marsaglia2003} in
which  unsigned longs  (64 bits)  have been  replaced by  unsigned  integers (32
bits).

This version  can also pass the whole {\it BigCrush} battery of tests.

\begin{algorithm}

\KwIn{InternalVarXorLikeArray: array with internal variables of 1 xor-like PRNGs
in global memory\;
NumThreads: Number of threads\;
array\_comb1, array\_comb2: Arrays containing combinations of size combination\_size\;}

\KwOut{NewNb: array containing random numbers in global memory}
\If{threadId is concerned} {
  retrieve data from InternalVarXorLikeArray[threadId] in local variables including shared memory and x\;
  offset = threadIdx\%combination\_size\;
  o1 = threadIdx-offset+array\_comb1[offset]\;
  o2 = threadIdx-offset+array\_comb2[offset]\;
  \For{i=1 to n} {
    t=xor-like()\;
    t=t\textasciicircum shmem[o1]\textasciicircum shmem[o2]\;
    shared\_mem[threadId]=t\;
    x = x\textasciicircum t\;

    store the new PRNG in NewNb[NumThreads*threadId+i]\;
  }
  store internal variables in InternalVarXorLikeArray[threadId]\;
}

\caption{Main kernel for the chaotic iterations based PRNG GPU efficient
version\label{IR}}
\label{algo:gpu_kernel2} 
\end{algorithm}

\subsection{Theoretical Evaluation of the Improved Version}

A run of Algorithm~\ref{algo:gpu_kernel2} consists in an operation ($x=x\oplus t$) having 
the form of Equation~\ref{equation Oplus}, which is equivalent to the iterative
system of Eq.~\ref{eq:generalIC}. That is, an iteration of the general chaotic
iterations is realized between the last stored value $x$ of the thread and a strategy $t$
(obtained by a bitwise exclusive or between a value provided by a xor-like() call
and two values previously obtained by two other threads).
To be certain that we are in the framework of Theorem~\ref{t:chaos des general},
we must guarantee that this dynamical system iterates on the space 
$\mathcal{X} = \mathcal{P}\left(\llbracket 1, \mathsf{N} \rrbracket\right)^\mathds{N}\times\mathds{B}^\mathsf{N}$.
The left term $x$ obviously belongs to $\mathds{B}^ \mathsf{N}$.
To prevent from any flaws of chaotic properties, we must check that the right 
term (the last $t$), corresponding to the strategies,  can possibly be equal to any
integer of $\llbracket 1, \mathsf{N} \rrbracket$. 

Such a result is obvious, as for the xor-like(), all the
integers belonging into its interval of definition can occur at each iteration, and thus the 
last $t$ respects the requirement. Furthermore, it is possible to
prove by an immediate mathematical induction that, as the initial $x$
is uniformly distributed (it is provided by a cryptographically secure PRNG),
the two other stored values shmem[o1] and shmem[o2] are uniformly distributed too,
(this is the induction hypothesis), and thus the next $x$ is finally uniformly distributed.

Thus Algorithm~\ref{algo:gpu_kernel2} is a concrete realization of the general
chaotic iterations presented previously, and for this reason, it satisfies the 
Devaney's formulation of a chaotic behavior.

\section{Experiments}
\label{sec:experiments}

Different experiments  have been  performed in order  to measure  the generation
speed. We have used a first computer equipped with a Tesla C1060 NVidia  GPU card
and an
Intel  Xeon E5530 cadenced  at 2.40  GHz,  and 
a second computer  equipped with a smaller  CPU and  a GeForce GTX  280. 
All the
cards have 240 cores.

In  Figure~\ref{fig:time_xorlike_gpu} we  compare the  quantity of  pseudorandom numbers
generated per second with various xor-like based PRNGs. In this figure, the optimized
versions use the {\it xor64} described in~\cite{Marsaglia2003}, whereas the naive versions
embed  the three  xor-like  PRNGs described  in Listing~\ref{algo:seqCIPRNG}.   In
order to obtain the optimal performances, the storage of pseudorandom numbers
into the GPU memory has been removed. This step is time consuming and slows down the numbers
generation.  Moreover this   storage  is  completely
useless, in case of applications that consume the pseudorandom
numbers  directly   after generation. We can see  that when the number of  threads is greater
than approximately 30,000 and lower than 5 million, the number of pseudorandom numbers generated
per second  is almost constant.  With the  naive version, this value ranges from 2.5 to
3GSamples/s.   With  the  optimized   version,  it  is  approximately  equal to
20GSamples/s. Finally  we can remark  that both GPU  cards are quite  similar, but in
practice,  the Tesla C1060  has more  memory than  the GTX  280, and  this memory
should be of better quality.
As a  comparison,   Listing~\ref{algo:seqCIPRNG}  leads   to the  generation of  about
138MSample/s when using one core of the Xeon E5530.

\begin{figure}[htbp]
\begin{center}
  \includegraphics[scale=.7]{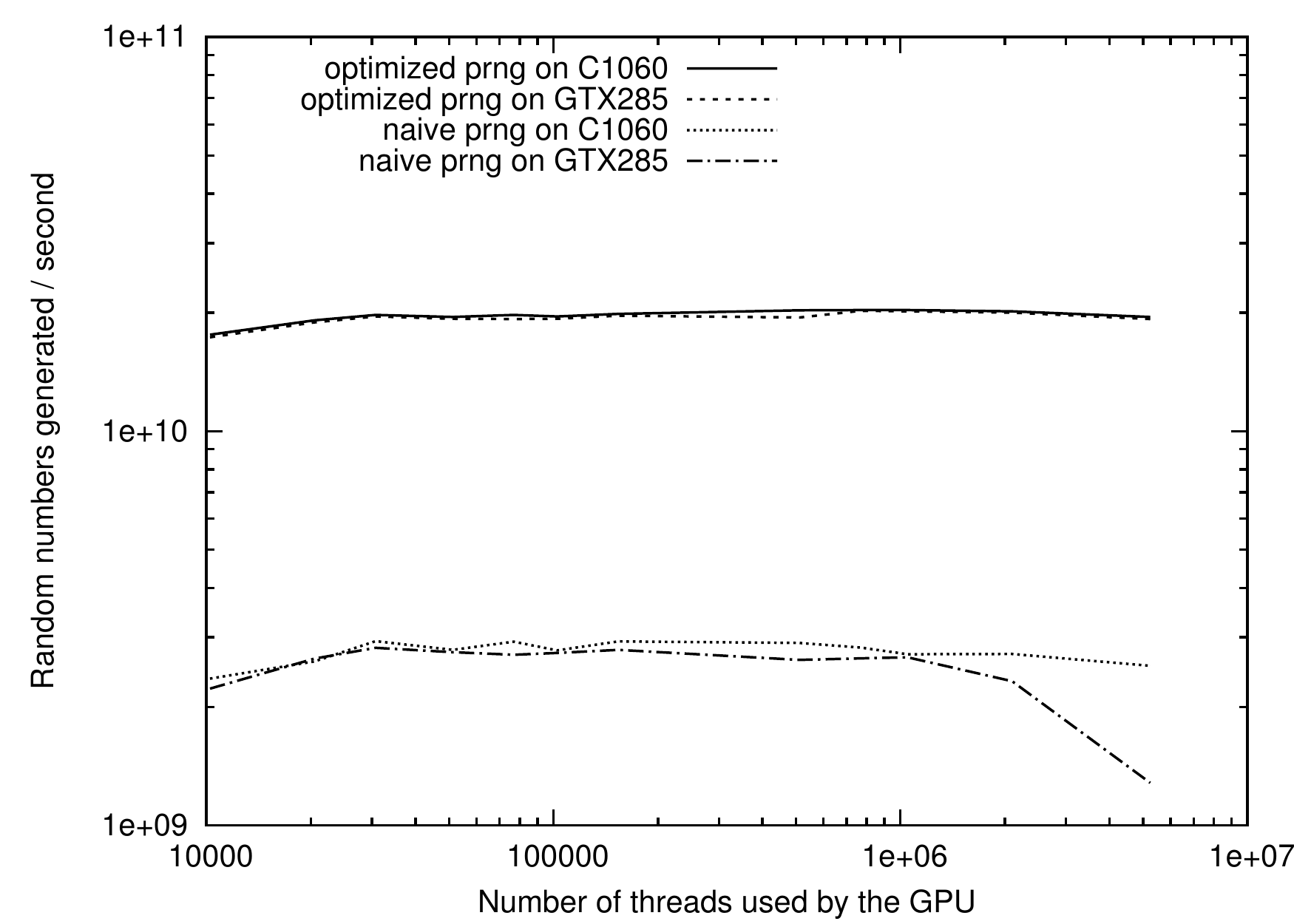}
\end{center}
\caption{Quantity of pseudorandom numbers generated per second with the xorlike-based PRNG}
\label{fig:time_xorlike_gpu}
\end{figure}

In Figure~\ref{fig:time_bbs_gpu} we highlight  the performances of the optimized
BBS-based PRNG on GPU.  On  the Tesla C1060 we obtain approximately 700MSample/s
and  on the  GTX 280  about  670MSample/s, which  is obviously  slower than  the
xorlike-based PRNG on GPU. However, we  will show in the next sections that this
new PRNG  has a strong  level of  security, which is  necessarily paid by  a speed
reduction.

\begin{figure}[htbp]
\begin{center}
  \includegraphics[scale=.7]{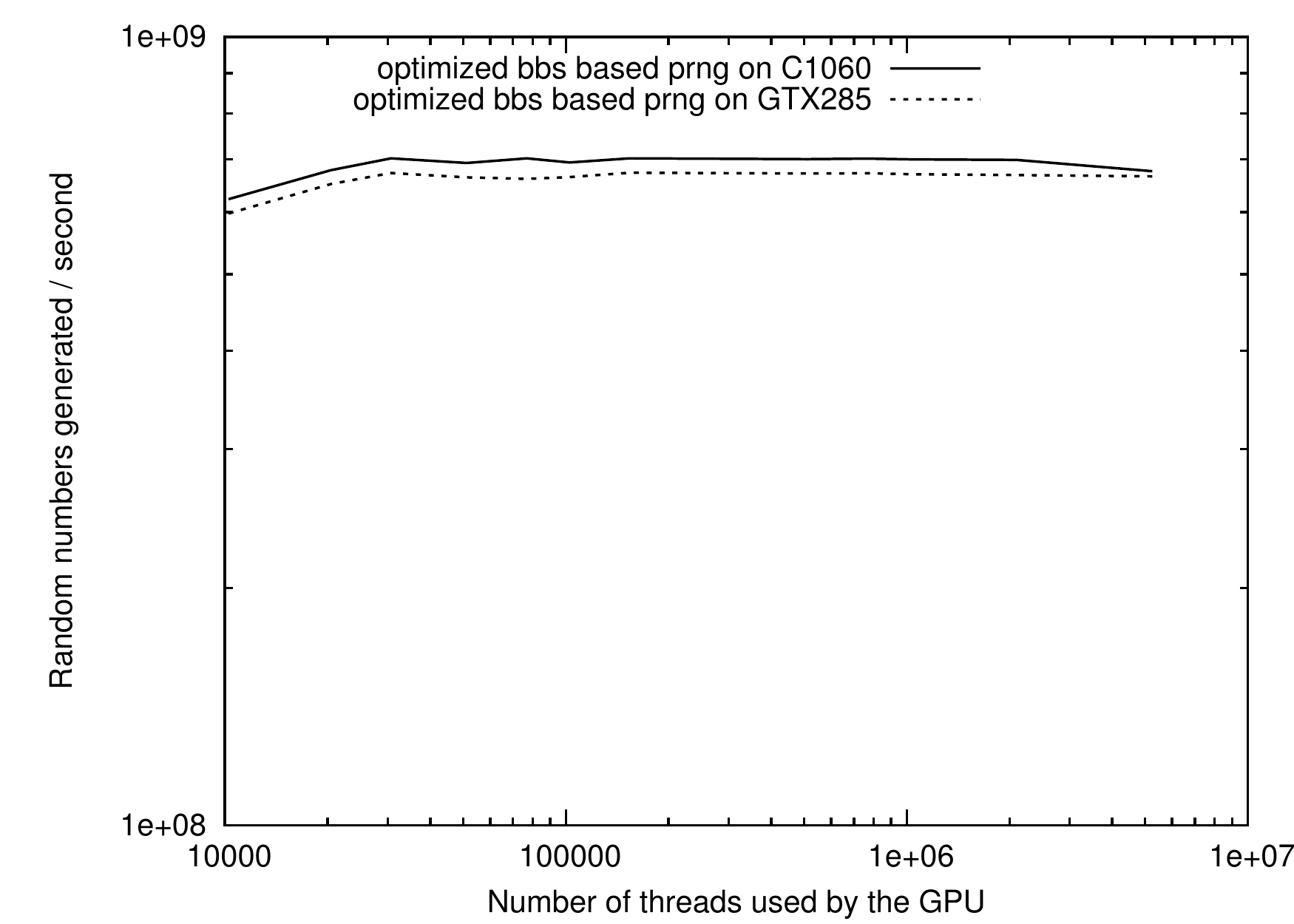}
\end{center}
\caption{Quantity of pseudorandom numbers generated per second using the BBS-based PRNG}
\label{fig:time_bbs_gpu}
\end{figure}

All  these  experiments allow  us  to conclude  that  it  is possible  to
generate a very large quantity of pseudorandom  numbers statistically perfect with the  xor-like version.
To a certain extend, it is also the case with the secure BBS-based version, the speed deflation being
explained by the fact that the former  version has ``only''
chaotic properties and statistical perfection, whereas the latter is also cryptographically secure,
as it is shown in the next sections.

\section{Security Analysis}
\label{sec:security analysis}

In this section the concatenation of two strings $u$ and $v$ is classically
denoted by $uv$.
In a cryptographic context, a pseudorandom generator is a deterministic
algorithm $G$ transforming strings  into strings and such that, for any
seed $s$ of length $m$, $G(s)$ (the output of $G$ on the input $s$) has size
$\ell_G(m)$ with $\ell_G(m)>m$.
The notion of {\it secure} PRNGs can now be defined as follows. 

\begin{definition}
A cryptographic PRNG $G$ is secure if for any probabilistic polynomial time
algorithm $D$, for any positive polynomial $p$, and for all sufficiently
large $m$'s,
$$| \mathrm{Pr}[D(G(U_m))=1]-Pr[D(U_{\ell_G(m)})=1]|< \frac{1}{p(m)},$$
where $U_r$ is the uniform distribution over $\{0,1\}^r$ and the
probabilities are taken over $U_m$, $U_{\ell_G(m)}$ as well as over the
internal coin tosses of $D$. 
\end{definition}

Intuitively, it means that there is no polynomial time algorithm that can
distinguish a perfect uniform random generator from $G$ with a non
negligible probability. The interested reader is referred
to~\cite[chapter~3]{Goldreich} for more information. Note that it is
quite easily possible to change the function $\ell$ into any polynomial
function $\ell^\prime$ satisfying $\ell^\prime(m)>m)$~\cite[Chapter 3.3]{Goldreich}.

The generation schema developed in (\ref{equation Oplus}) is based on a
pseudorandom generator. Let $H$ be a cryptographic PRNG. We may assume,
without loss of generality, that for any string $S_0$ of size $N$, the size
of $H(S_0)$ is $kN$, with $k>2$. It means that $\ell_H(N)=kN$. 
Let $S_1,\ldots,S_k$ be the 
strings of length $N$ such that $H(S_0)=S_1 \ldots S_k$ ($H(S_0)$ is the concatenation of
the $S_i$'s). The cryptographic PRNG $X$ defined in (\ref{equation Oplus})
is the algorithm mapping any string of length $2N$ $x_0S_0$ into the string
$(x_0\oplus S_0 \oplus S_1)(x_0\oplus S_0 \oplus S_1\oplus S_2)\ldots
(x_o\bigoplus_{i=0}^{i=k}S_i)$. One in particular has $\ell_{X}(2N)=kN=\ell_H(N)$. 
We claim now that if this PRNG is secure,
then the new one is secure too.

\begin{proposition}
\label{cryptopreuve}
If $H$ is a secure cryptographic PRNG, then $X$ is a secure cryptographic
PRNG too.
\end{proposition}

\begin{proof}
The proposition is proved by contraposition. Assume that $X$ is not
secure. By Definition, there exists a polynomial time probabilistic
algorithm $D$, a positive polynomial $p$, such that for all $k_0$ there exists
$N\geq \frac{k_0}{2}$ satisfying 
$$| \mathrm{Pr}[D(X(U_{2N}))=1]-\mathrm{Pr}[D(U_{kN}=1]|\geq \frac{1}{p(2N)}.$$
We describe a new probabilistic algorithm $D^\prime$ on an input $w$ of size
$kN$:
\begin{enumerate}
\item Decompose $w$ into $w=w_1\ldots w_{k}$, where each $w_i$ has size $N$.
\item Pick a string $y$ of size $N$ uniformly at random.
\item Compute $z=(y\oplus w_1)(y\oplus w_1\oplus w_2)\ldots (y
  \bigoplus_{i=1}^{i=k} w_i).$
\item Return $D(z)$.
\end{enumerate}

Consider  for each $y\in \mathbb{B}^{kN}$ the function $\varphi_{y}$
from $\mathbb{B}^{kN}$ into $\mathbb{B}^{kN}$ mapping $w=w_1\ldots w_k$
(each $w_i$ has length $N$) to 
$(y\oplus w_1)(y\oplus w_1\oplus w_2)\ldots (y
  \bigoplus_{i=1}^{i=k_1} w_i).$ By construction, one has for every $w$,
\begin{equation}\label{PCH-1}
D^\prime(w)=D(\varphi_y(w)),
\end{equation}
where $y$ is randomly generated. 
Moreover, for each $y$, $\varphi_{y}$ is injective: if 
$(y\oplus w_1)(y\oplus w_1\oplus w_2)\ldots (y\bigoplus_{i=1}^{i=k_1}
w_i)=(y\oplus w_1^\prime)(y\oplus w_1^\prime\oplus w_2^\prime)\ldots
(y\bigoplus_{i=1}^{i=k} w_i^\prime)$, then for every $1\leq j\leq k$,
$y\bigoplus_{i=1}^{i=j} w_i^\prime=y\bigoplus_{i=1}^{i=j} w_i$. It follows,
by a direct induction, that $w_i=w_i^\prime$. Furthermore, since $\mathbb{B}^{kN}$
is finite, each $\varphi_y$ is bijective. Therefore, and using (\ref{PCH-1}),
one has
\begin{equation}\label{PCH-2}
\mathrm{Pr}[D^\prime(U_{kN})=1]=\mathrm{Pr}[D(\varphi_y(U_{kN}))=1]=\mathrm{Pr}[D(U_{kN})=1].
\end{equation}

Now, using (\ref{PCH-1}) again, one has  for every $x$,
\begin{equation}\label{PCH-3}
D^\prime(H(x))=D(\varphi_y(H(x))),
\end{equation}
where $y$ is randomly generated. By construction, $\varphi_y(H(x))=X(yx)$,
thus
\begin{equation}\label{PCH-3}
D^\prime(H(x))=D(yx),
\end{equation}
where $y$ is randomly generated. 
It follows that 

\begin{equation}\label{PCH-4}
\mathrm{Pr}[D^\prime(H(U_{N}))=1]=\mathrm{Pr}[D(U_{2N})=1].
\end{equation}
 From (\ref{PCH-2}) and (\ref{PCH-4}), one can deduce that
there exists a polynomial time probabilistic
algorithm $D^\prime$, a positive polynomial $p$, such that for all $k_0$ there exists
$N\geq \frac{k_0}{2}$ satisfying 
$$| \mathrm{Pr}[D(H(U_{N}))=1]-\mathrm{Pr}[D(U_{kN}=1]|\geq \frac{1}{p(2N)},$$
proving that $H$ is not secure, which is a contradiction. 
\end{proof}

\section{Cryptographical Applications}

\subsection{A Cryptographically Secure PRNG for GPU}
\label{sec:CSGPU}

It is  possible to build a  cryptographically secure PRNG based  on the previous
algorithm (Algorithm~\ref{algo:gpu_kernel2}).   Due to Proposition~\ref{cryptopreuve},
it simply consists  in replacing
the  {\it  xor-like} PRNG  by  a  cryptographically  secure one.  
We have chosen the Blum Blum Shum generator~\cite{BBS} (usually denoted by BBS) having the form:
$$x_{n+1}=x_n^2~ mod~ M$$  where $M$ is the product of  two prime numbers (these
prime numbers  need to be congruent  to 3 modulus  4). BBS is known to be
very slow and only usable for cryptographic applications.

The modulus operation is the most time consuming operation for current
GPU cards.  So in order to obtain quite reasonable performances, it is
required to use only modulus  on 32-bits integer numbers. Consequently
$x_n^2$ need  to be lesser than $2^{32}$,  and thus the number $M$ must be
lesser than $2^{16}$.  So in practice we can choose prime numbers around
256 that are congruent to 3 modulus 4.  With 32-bits numbers, only the
4 least significant bits of $x_n$ can be chosen (the maximum number of
indistinguishable    bits    is    lesser    than   or    equals    to
$log_2(log_2(M))$). In other words, to generate a  32-bits number, we need to use
8 times  the BBS  algorithm with possibly different  combinations of  $M$. This
approach is  not sufficient to be able to pass  all the tests of TestU01,
as small values of  $M$ for the BBS  lead to
  small periods. So, in  order to add randomness  we have proceeded with
the followings  modifications. 
\begin{itemize}
\item
Firstly, we  define 16 arrangement arrays  instead of 2  (as described in
Algorithm \ref{algo:gpu_kernel2}), but only 2 of them are used at each call of
the  PRNG kernels. In  practice, the  selection of   combination
arrays to be used is different for all the threads. It is determined
by using  the three last bits  of two internal variables  used by BBS.
In Algorithm~\ref{algo:bbs_gpu},
character  \& is for the  bitwise AND. Thus using  \&7 with  a number
gives the last 3 bits, thus providing a number between 0 and 7.
\item
Secondly, after the  generation of the 8 BBS numbers  for each thread, we
have a 32-bits number whose period is possibly quite small. So
to add randomness,  we generate 4 more BBS numbers   to
shift  the 32-bits  numbers, and  add up to  6 new  bits.  This  improvement is
described  in Algorithm~\ref{algo:bbs_gpu}.  In  practice, the last 2 bits
of the first new BBS number are  used to make a left shift of at most
3 bits. The  last 3 bits of the  second new BBS number are  added to the
strategy whatever the value of the first left shift. The third and the
fourth new BBS  numbers are used similarly to apply  a new left shift
and add 3 new bits.
\item
Finally, as  we use 8 BBS numbers  for each thread, the  storage of these
numbers at the end of the  kernel is performed using a rotation. So,
internal  variable for  BBS number  1 is  stored in  place  2, internal
variable  for BBS  number 2  is  stored in  place 3,  ..., and finally, internal
variable for BBS number 8 is stored in place 1.
\end{itemize}

\begin{algorithm}

\KwIn{InternalVarBBSArray: array with internal variables of the 8 BBS
in global memory\;
NumThreads: Number of threads\;
array\_comb: 2D Arrays containing 16 combinations (in first dimension)  of size combination\_size (in second dimension)\;
array\_shift[4]=\{0,1,3,7\}\;
}

\KwOut{NewNb: array containing random numbers in global memory}
\If{threadId is concerned} {
  retrieve data from InternalVarBBSArray[threadId] in local variables including shared memory and x\;
  we consider that bbs1 ... bbs8 represent the internal states of the 8 BBS numbers\;
  offset = threadIdx\%combination\_size\;
  o1 = threadIdx-offset+array\_comb[bbs1\&7][offset]\;
  o2 = threadIdx-offset+array\_comb[8+bbs2\&7][offset]\;
  \For{i=1 to n} {
    t$<<$=4\;
    t|=BBS1(bbs1)\&15\;
    ...\;
    t$<<$=4\;
    t|=BBS8(bbs8)\&15\;
    \tcp{two new shifts}
    shift=BBS3(bbs3)\&3\;
    t$<<$=shift\;
    t|=BBS1(bbs1)\&array\_shift[shift]\;
    shift=BBS7(bbs7)\&3\;
    t$<<$=shift\;
    t|=BBS2(bbs2)\&array\_shift[shift]\;
    t=t\textasciicircum  shmem[o1]\textasciicircum     shmem[o2]\;
    shared\_mem[threadId]=t\;
    x = x\textasciicircum   t\;

    store the new PRNG in NewNb[NumThreads*threadId+i]\;
  }
  store internal variables in InternalVarXorLikeArray[threadId] using a rotation\;
}

\caption{main kernel for the BBS based PRNG GPU}
\label{algo:bbs_gpu}
\end{algorithm}

In Algorithm~\ref{algo:bbs_gpu}, $n$ is for  the quantity of random numbers that
a thread has to  generate.  The operation t<<=4 performs a left  shift of 4 bits
on the variable  $t$ and stores the result in  $t$, and $BBS1(bbs1)\&15$ selects
the last  four bits  of the  result of $BBS1$.   Thus an  operation of  the form
$t<<=4; t|=BBS1(bbs1)\&15\;$  realizes in $t$ a  left shift of 4  bits, and then
puts the 4 last bits of $BBS1(bbs1)$  in the four last positions of $t$.  Let us
remark that the initialization $t$ is not a  necessity as we fill it 4 bits by 4
bits, until  having obtained 32-bits.  The  two last new shifts  are realized in
order to enlarge the small periods of  the BBS used here, to introduce a kind of
variability.  In these operations, we make twice a left shift of $t$ of \emph{at
  most}  3 bits,  represented by  \texttt{shift} in  the algorithm,  and  we put
\emph{exactly} the \texttt{shift}  last bits from a BBS  into the \texttt{shift}
last bits of $t$. For this, an array named \texttt{array\_shift}, containing the
correspondence between the  shift and the number obtained  with \texttt{shift} 1
to make the \texttt{and} operation is used. For example, with a left shift of 0,
we  make an  and operation  with 0,  with  a left  shift of  3, we  make an  and
operation with 7 (represented by 111 in binary mode).

It should  be noticed that this generator has once more the form $x^{n+1} = x^n \oplus S^n$,
where $S^n$ is referred in this algorithm as $t$: each iteration of this
PRNG ends with $x = x \wedge t$. This $S^n$ is only constituted
by secure bits produced by the BBS generator, and thus, due to
Proposition~\ref{cryptopreuve}, the resulted PRNG is cryptographically
secure.

\subsection{Toward a Cryptographically Secure and Chaotic Asymmetric Cryptosystem}
\label{Blum-Goldwasser}
We finish this research work by giving some thoughts about the use of
the proposed PRNG in an asymmetric cryptosystem.
This first approach will be further investigated in a future work.

\subsubsection{Recalls of the Blum-Goldwasser Probabilistic Cryptosystem}

The Blum-Goldwasser cryptosystem is a cryptographically secure asymmetric key encryption algorithm 
proposed in 1984~\cite{Blum:1985:EPP:19478.19501}.  The encryption algorithm 
implements a XOR-based stream cipher using the BBS PRNG, in order to generate 
the keystream. Decryption is done by obtaining the initial seed thanks to
the final state of the BBS generator and the secret key, thus leading to the
 reconstruction of the keystream.

The key generation consists in generating two prime numbers $(p,q)$, 
randomly and independently of each other, that are
 congruent to 3 mod 4, and to compute the modulus $N=pq$.
The public key is $N$, whereas the secret key is the factorization $(p,q)$.

Suppose Bob wishes to send a string $m=(m_0, \dots, m_{L-1})$ of $L$ bits to Alice:
\begin{enumerate}
\item Bob picks an integer $r$ randomly in the interval $\llbracket 1,N\rrbracket$ and computes $x_0 = r^2~mod~N$.
\item He uses the BBS to generate the keystream of $L$ pseudorandom bits $(b_0, \dots, b_{L-1})$, as follows. For $i=0$ to $L-1$,
\begin{itemize}
\item $i=0$.
\item While $i \leqslant L-1$:
\begin{itemize}
\item Set $b_i$ equal to the least-significant\footnote{As signaled previously, BBS can securely output up to $\mathsf{N} = \lfloor log(log(N)) \rfloor$ of the least-significant bits of $x_i$ during each round.} bit of $x_i$,
\item $i=i+1$,
\item $x_i = (x_{i-1})^2~mod~N.$
\end{itemize}
\end{itemize}
\item The ciphertext is computed by XORing the plaintext bits $m$ with the keystream: $ c = (c_0, \dots, c_{L-1}) = m \oplus  b$. This ciphertext is $[c, y]$, where $y=x_{0}^{2^{L}}~mod~N.$
\end{enumerate}

When Alice receives $\left[(c_0, \dots, c_{L-1}), y\right]$, she can recover $m$ as follows:
\begin{enumerate}
\item Using the secret key $(p,q)$, she computes $r_p = y^{((p+1)/4)^{L}}~mod~p$ and $r_q = y^{((q+1)/4)^{L}}~mod~q$.
\item The initial seed can be obtained using the following procedure: $x_0=q(q^{-1}~{mod}~p)r_p + p(p^{-1}~{mod}~q)r_q~{mod}~N$.
\item She recomputes the bit-vector $b$ by using BBS and $x_0$.
\item Alice finally computes the plaintext by XORing the keystream with the ciphertext: $ m = c \oplus  b$.
\end{enumerate}

\subsubsection{Proposal of a new Asymmetric Cryptosystem Adapted from Blum-Goldwasser}

We propose to adapt the Blum-Goldwasser protocol as follows. 
Let $\mathsf{N} = \lfloor log(log(N)) \rfloor$ be the number of bits that can
be obtained securely with the BBS generator using the public key $N$ of Alice.
Alice will pick randomly $S^0$ in $\llbracket 0, 2^{\mathsf{N}-1}\rrbracket$ too, and
her new public key will be $(S^0, N)$.

To encrypt his message, Bob will compute
\begin{equation}
c = \left(m_0 \oplus (b_0 \oplus S^0), m_1 \oplus (b_0 \oplus b_1 \oplus S^0), \hdots, m_{L-1} \oplus (b_0 \oplus b_1 \hdots \oplus b_{L-1} \oplus S^0) \right)
\end{equation}
instead of $\left(m_0 \oplus b_0, m_1 \oplus b_1, \hdots, m_{L-1} \oplus b_{L-1} \right)$. 

The same decryption stage as in Blum-Goldwasser leads to the sequence 
$\left(m_0 \oplus S^0, m_1 \oplus S^0, \hdots, m_{L-1} \oplus S^0 \right)$.
Thus, with a simple use of $S^0$, Alice can obtain the plaintext.
By doing so, the proposed generator is used in place of BBS, leading to
the inheritance of all the properties presented in this paper.

\section{Conclusion}

In  this  paper, a formerly proposed PRNG based on chaotic iterations
has been generalized to improve its speed. It has been proven to be
chaotic according to Devaney.
Efficient implementations on  GPU using xor-like  PRNGs as input generators
have shown that a very large quantity of pseudorandom numbers can be generated per second (about
20Gsamples/s), and that these proposed PRNGs succeed to pass the hardest battery in TestU01,
namely the BigCrush.
Furthermore, we have shown that when the inputted generator is cryptographically
secure, then it is the case too for the PRNG we propose, thus leading to
the possibility to develop fast and secure PRNGs using the GPU architecture.
Thoughts about an improvement of the Blum-Goldwasser cryptosystem, using the 
proposed method, has been finally proposed.

In future  work we plan to extend these researches, building a parallel PRNG for  clusters or
grid computing. Topological properties of the various proposed generators will be investigated,
and the use of other categories of PRNGs as input will be studied too. The improvement
of Blum-Goldwasser will be deepened. Finally, we
will try to enlarge the quantity of pseudorandom numbers generated per second either
in a simulation context or in a cryptographic one.

\bibliographystyle{plain} 
\bibliography{mabase}
\end{document}